\newcommand{\changed}[1]{#1}
\newcommand{\changedb}[1]{#1}
\newcommand{\changedc}[1]{#1}

\documentclass[twocolumn,amsthm,secthm]{autart}    

\usepackage{graphicx}          

\usepackage{subfig}
\usepackage{amsmath}
\usepackage{amsfonts}
\usepackage{amssymb}  
\usepackage{bm}
\usepackage{siunitx}
\usepackage{xspace}
\usepackage{xfrac}
\usepackage{empheq}
\usepackage{enumerate}
\usepackage{xcolor}
\usepackage{url}
\usepackage[round]{natbib}

\DeclareMathOperator{\sign}{sign}

\begin{document}

\begin{frontmatter}

\title{Proper Implicit Discretization of the Super-Twisting Controller---without and with Actuator Saturation}
\author[IRTCD,IRT]{Richard Seeber}\ead{richard.seeber@tugraz.at},
\author[IRT]{Benedikt Andritsch}\ead{benedikt.andritsch@tugraz.at}\thanks{\begingroup\color{red}
This is the accepted manuscript version of the paper: R. Seeber, B. Andritsch. ``Proper implicit discretization of the super-twisting controller---without and with actuator saturation'', Automatica 173, Article No. 112027, 2025, available as open access at \url{https://dx.doi.org/10.1016/j.automatica.2024.112027}. \textbf{Please cite the publisher's version.}\endgroup}
    \thanks{The financial support by the Christian Doppler Research Association, the Austrian Federal Ministry for Digital and Economic Affairs and the National Foundation for Research, Technology and Development is gratefully acknowledged.}\address[IRTCD]{Christian Doppler Laboratory for Model Based Control of Complex Test Bed Systems, Institute of Automation and Control, Graz University of Technology, Graz, Austria}\address[IRT]{Institute of Automation and Control, Graz University of Technology, Graz, Austria}


\begin{abstract}
The discrete-time implementation of the super-twisting sliding mode controller for a plant with disturbances with bounded slope, zero-order hold actuation, and actuator constraints is considered.
Motivated by restrictions of existing implicit or semi-implicit discretization variants, a new proper implicit discretization for the super-twisting controller is proposed.
This discretization is then extended to the conditioned super-twisting controller, which mitigates windup in presence of actuator constraints by means of the conditioning technique.
It is proven that the proposed controllers achieve best possible worst-case performance subject to similarly simple stability conditions as their continuous-time counterparts.
Numerical simulations and comparisons demonstrate and illustrate the results.
 \end{abstract}

\end{frontmatter}

\begingroup
\def\QED{$\blacksquare$}
\let\oldepsilon\epsilon
\let\oldphi\phi
\renewcommand{\epsilon}{\varepsilon}
\renewcommand{\phi}{\varphi}
\newcommand{\ee}{\mathrm{e}}

\newcommand{\abs}[1]{\left\lvert #1 \right\rvert}

\newcommand{\lambdamin}{\lambda\sbrm{min}}
\newcommand{\lambdamax}{\lambda\sbrm{max}}
\newcommand{\Ki}{\K\sbrm{I}}
\newcommand{\ki}{k\sbrm{I}}
\def\us{\u^{*}}
\newcommand{\wpmax}{L}
\newcommand{\wmax}{W}
\newcommand{\umax}{U}

\newcommand{\diffd}{\mathrm{d}}
\newcommand{\dt}[1]{\deriv{#1}{t}}
\newcommand{\deriv}[2]{\frac{\diffd {#1}}{\diffd  {#2}}}
\newcommand{\derivk}[3]{\frac{\diffd^{#3} {#1}}{\diffd  {#2}^{#3}}}
\newcommand{\pderiv}[2]{\frac{\partial{#1}}{\partial{#2}}}
\newcommand{\pderivs}[3]{\frac{\partial^2{#1}}{\partial{#2}\partial{#3}}}
\newcommand{\pderivk}[3]{\frac{\partial^{#3} {#1}}{\partial{#2}^{#3}}}

\newcommand{\ceil}[1]{\left\lceil{#1}\right\rceil}
\newcommand{\floor}[1]{\left\lfloor{#1}\right\rfloor}
\newcommand{\norm}[2][]{\left\lVert #2 \right\rVert_{#1}}

\newcommand{\spow}[2]{\left\lfloor #1 \right\rceil^{#2}}
\newcommand{\spowf}[3]{\spow{#1}{\frac{#2}{#3}}}
\newcommand{\apow}[2]{\abs{#1}^{#2}}
\newcommand{\apowf}[3]{\apow{#1}{\frac{#2}{#3}}}

\newcommand{\integ}[4]{\int_{#2}^{#3}{{#4} \, \diffd {#1}}}
\newcommand{\integeq}[4]{\int_{{#1}={#2}}^{#3}{{#4} \, \diffd {#1}}}

\newcommand{\sbrm}[1]{\sb{\mathrm{#1}}}

\newcommand{\TT}{^{\mathrm{T}}}
\newcommand{\HH}{^{\mathrm{H}}}

\def\x{\mathbf{x}}
\def\h{\mathbf{h}}
\newcommand{\vzeta}{\bm{\zeta}}
\newcommand{\vxi}{\bm{\xi}}

\newcommand{\RR}{\mathbb{R}}
\newcommand{\CC}{\mathbb{C}}
 \newcommand{\NN}{\mathbb{N}}
\newcommand{\ZZ}{\mathbb{Z}}
\newcommand{\mex}{\hfill$\triangle$}
\newcommand{\sat}{\operatorname{sat}}
\renewcommand{\F}{\mathcal{F}}
\section{Introduction}

Sliding mode control \changed{(SMC)} is a 
robust control technique for systems that contain uncertainties.
In continuous time, \changed{SMC} manages to completely reject disturbances that fulfill certain requirements after a finite convergence time, see e.g.~\citep{shtessel2013}.
However, implementing \changed{SMC} in practice is not straightforward and often leads to chattering due to discretization and measurement noise, as shown by~\cite{levant2011}.
One approach to reduce chattering effects is higher-order \changed{SMC}. A popular second-order sliding mode controller that rejects disturbances with bounded derivative is the Super-Twisting Controller (STC) introduced by~\cite{levant1993sliding}.

In~\citep{koch2019a} the authors \changed{proposed} a discretization of the super-twisting algorithm, i.e., the closed-loop system with the STC, on the basis of an eigenvalue mapping from continuous-
to discrete-time. They \changed{applied} an implicit and an exact
mapping yielding an implicit and a matching discretization, respectively, mitigating chattering to some extent.
\cite{hanan2021} \changed{proposed} a low-chattering discretization of \changed{SMC} that \changed{is based on} an explicit discretization and significantly reduces chattering effects.
An implicit discretization of the STC that avoids discretization chattering based on ideas from~\citep{acary2010} was proposed by~\cite{brogliato2020implicit}.
Another discrete-time representation of the STC that is based on a semi-implicit discretization and also avoids discretization chattering was developed by~\cite{xiong2022discrete}.
More recently, a modified implicitly discretized STC was proposed by~\cite{andritsch2024modified}, where the authors also performed detailed comparisons between existing discretizations of the STC. 

Apart from the need for a discrete-time implementation, real-life plants in practice also often have limitations regarding the control input that can be applied. The STC includes controller dynamics, which in combination with a saturated control input can lead to windup effects in the controller. This may diminish the control performance, e.g. by increased convergence times or large overshoot of the system states. A method to avoid windup is the so-called conditioning technique by~\cite{hanus1987conditioning}.
For the case of saturated control, the conditioning technique was applied to the continuous-time STC by~\cite{seeber2020conditioned}.
\changed{\cite{Reichhartinger_2023} applied anti-windup schemes directly to discrete-time realizations of the STC, e.g. \citep{koch2019a}.}
In~\cite{yang2022semi} the authors applied the semi-implicit discretization by~\cite{xiong2022discrete} to  \changed{the}  conditioned STC, obtaining a discrete-time implementation of the STC with windup mitigation. 

Compared to the continuous-time STC, the discussed discrete-time implementations without and with actuator saturation exhibit various restrictions.
In particular, the discrete-time STC by~\cite{brogliato2020implicit} can handle a smaller class of disturbances; this will be formally demonstrated later on and is also shown in the comparative results in~\citep{andritsch2024modified}.
The latter comparisons also show that the discretization by~\cite{xiong2022discrete} is harder to tune, because certain parameter selections result in a significantly increased convergence time.
These shortcomings are avoided by~\cite{andritsch2024modified}; however, their discretization lacks a proof of global \changed{closed-loop} stability in the presence of a disturbance, \changed{as do the discrete-time controllers with saturation studied in \citep{Reichhartinger_2023}}.
The conditioned discrete-time STC by~\cite{yang2022semi}, finally, has similar drawbacks as the unconditioned semi-implicitly discretized STC by~\cite{xiong2022discrete} and additionally does not necessarily achieve the best possible worst-case error, as shown in~\citep{seeber2023discussion}.

The present paper derives a novel discretization of the STC that does not have the disadvantages of previously proposed discretizations and is shown to yield the best possible worst-case control error. Furthermore, a complete stability proof and simple stability conditions are provided, extending those given by~\cite{brogliato2020implicit}, in addition to extending the class of perturbations.
For the case of saturated actuation, the conditioning technique is applied to the proposed discrete-time STC, yielding an implicit discretization of the conditioned STC.
Also for this case, stability conditions are derived that are very similar to those obtained in~\citep{seeber2020conditioned} for the continuous-time case.
\changed{Moreover, explicit controller realizations are derived, in a similar fashion as recently noted by \cite{brogliato2023comments} for the first-order controllers by \cite{haddad2020finite}.}

The paper is structured as follows.
Section~\ref{sec:problem} introduces and motivates the problem statement based on existing approaches in literature.
\changed{Sections~\ref{sec:implicit} and~\ref{sec:conditioned} present the main results: implementations and formal guarantees for the implicit STC in absence and presence of actuator saturation, respectively.}
\changed{Sections~\ref{sec:explicit} and~\ref{sec:stability} then present the corresponding derivations, with the derivation of explicit control laws being contained in the former and the stability analysis being performed in the latter.}
Section~\ref{sec:simulation} demonstrates the effectiveness of the proposed discrete-time controllers by means of simulations and \changed{comparisons}.
Section~\ref{sec:conclusion}, finally, concludes the paper.

\textbf{Notation:}
$\RR$, $\RR_{\ge 0}$, $\RR_{>0}$, $\ZZ$, $\NN$, and $\NN_0$ denote reals, nonnegative reals, positive reals, integers, positive integers, and nonnegative integers.
\changedc{For $M \in \RR_{\ge 0}$, define the saturation function $\sat_M : \RR \to [-M,M]$ as $\sat_M(y) = \max\{-M,\min\{M,y\}\}$.} \changedc{The scalar-valued sign function with $\sign(y) = \frac{y}{|y|}$ for $y\neq 0$ and $\sign(0) = 0$ is used, where $|y|$ denotes the absolute value of $y$.}
For $y,p \in \RR$, $p\ne0$, the abbreviation $\spow{y}{p} = \apow{y}{p} \sign(y)$ is used, and $\spow{y}{0}$ denotes the set-valued sign function defined as $\spow{y}{0} = \{ \sign(y) \}$ for $y \ne 0$ and $\spow{0}{0} = [-1,1]$. The real-valued mod-operator $a \text{ mod } b = r$, where $a, b \in \RR$, $b \ne 0$,  \changed{is the unique $r \in [0,|b|)$} fulfilling $a = r + k b$ with $k \in \mathbb{Z}$.
\changedb{The set of all subsets of a set $\mathcal{S}$ is denoted by $2^{\mathcal{S}}$.}
\changedc{Recall that a scalar-valued function $g : \RR^n \to \RR$ is called upper-semicontinuous, if $\limsup_{\vxi \to \x} g(\vxi) \le g(\x)$ holds for all $\x \in \RR^n$, and a set-valued function $\mathcal{F} : \RR^n \to 2^{\RR^n}$ is called upper-semicontinuous, if $\lim_{\vxi \to \x} \sup_{\vzeta_1 \in \mathcal{F}(\vxi)} \inf_{\vzeta_2 \in \mathcal{F}(\x)} \norm{\vzeta_1 - \vzeta_2} = 0$ holds for all $\x \in \RR^n$, cf. e.g. \cite{polyakov2014stability}.}

\section{Motivation and Problem Statement}
\label{sec:problem}

\subsection{Zero-Order Hold Sampled Sliding Mode Control}

Consider a \changed{scalar} sliding variable \changed{$x \in \RR$} governed by
\begin{equation}
    \label{eq:plant}
    \dot x = u + w
\end{equation}
with a control input $u : \RR_{\ge 0} \to \RR$ generated by a zero-order hold element with sampling time $T \in \RR_{> 0}$ from a control input sequence $(u_k)$\changed{, $k \in \NN_0$,} according to
\begin{equation}
    \label{eq:zoh}
    u(t) = u_k \qquad \text{for $t \in [kT, (k+1)T)$}
\end{equation}
and an unknown disturbance $w : \RR_{\ge 0} \to \RR$ whose slope and amplitude are bounded by $|\dot w(t)| \le L \in \RR_{\ge 0}$ almost everywhere and $|w| \le W \in \RR_{\ge 0} \cup \{ \infty \}$ for all $t \in \RR_{\ge 0}$.
The control goal is to drive the sliding variable to a vicinity of the origin that is as small as possible, considering that the disturbance is unknown.
To that end, \changed{consider the samples} $x_k = x(kT)$ \changed{with $k \in \NN_0$} and define 
\begin{equation}
    \label{eq:wk}
    w_k = \frac{1}{T} \int_{kT}^{(k+1)T} w(\tau) \, \diffd \tau
\end{equation}
to obtain the zero-order hold discretization of \eqref{eq:plant} as
\begin{equation}
    \label{eq:plant:discrete}
    x_{k+1} = x_k + T (u_k + w_k).
\end{equation}
It is easy to verify that the disturbance $w_k$ therein satisfies $|w_{k+1} - w_{k}| \le LT$ and $|w_k| \le W$ for all $k \in \NN_0$.

The following motivating proposition shows a lower bound on the worst-case disturbance rejection.
It is proven in \changed{Appendix~\ref{proof-prop-best}}.
\begin{prop}
    \label{prop:best}
    Let $L \in \RR_{\ge 0}$, $T \in \RR_{> 0}$, $W \ge \frac{3 LT}{2}$ and consider the continuous-time plant \eqref{eq:plant} with zero-order hold input \eqref{eq:zoh} and sampled output $x_k = x(kT)$.
    Then, for every initial condition $x_0 \in \RR$, for every causal control law, i.e., for every sequence $(h_k)$ of functions $h_k : \RR^{k+1} \to \RR$ such that
    $
        u_k = h_k(x_k, x_{k-1}, \ldots, x_0),
        $
    and for every $K \in \NN$ there exists a Lipschitz continuous disturbance $w : \RR_{\ge 0} \to \RR$ satisfying $|w(t)| \le W$ and $|\dot w(t)| \le L$ almost everywhere such that
    \begin{equation}
        \changed{\sup_{t \ge KT} |x(t)| \ge} \sup_{k \ge K} |x_k| \ge L T^2
    \end{equation}
    holds for the corresponding closed-loop trajectory.
\end{prop}
\begin{figure}
	\centering
	\includegraphics[width=\linewidth]{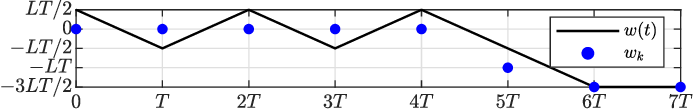}
	\caption{Example for a disturbance signal $w(t)$ and corresponsing $w_k$ limiting the worst-case error as in Proposition~\ref{prop:best}.}\label{fig:proposition:disturbance}
\end{figure}
\begin{rem}
    Fig.~\ref{fig:proposition:disturbance} shows \changed{a disturbance} $w(t)$ obtained from the proof \changed{of Proposition~\ref{prop:best}} with $K=4$ and corresponding $w_k$ that \changed{leads} to the state $x_k$ reaching the best possible worst-case bound $L T^2$.

    In \emph{continuous time}, a well-known sliding mode controller \changed{for the plant \eqref{eq:plant}} is the super-twisting controller
\begin{align}
    \label{eq:sta}
    u &= -k_1 \spowf{x}{1}{2} + v, &
    \dot v &= -k_2 \sign(x)
    \end{align}
proposed by \cite{levant1993sliding}, with positive parameters $k_1, k_2$ and trajectories understood in the sense of \cite{filippov1988differential}.
The goal of the present paper is to obtain a discrete-time implementation of this controller for the sampled control problem such that
\begin{itemize}
    \item 
        the optimal worst-case performance from Proposition~\ref{prop:best} is attained in finite time, i.e., such that \changed{$|x(t)| \le LT^2$ holds after a finite time}, and
    \item
        this optimal performance is maintained also in the presence of actuator saturation while, additionally, controller windup is mitigated.
\end{itemize}

Arguably, the most promising approaches for achieving these goals are the implicit or semi-implicit discretization techniques due to \cite{brogliato2020implicit,xiong2022discrete} in combination with the conditioning technique for windup mitigation proposed by \cite{hanus1987conditioning}.
In the following, the state-of-the-art solutions in that regard are discussed to motivate the present work.

\subsection{State-of-the-Art Implicit Super-Twisting Control}

\cite{brogliato2020implicit} propose an implicit super-twisting controller
given by the generalized equations
\begin{subequations}
    \label{eq:impl0}
    \begin{align}
    \label{eq:impl0:u}
        u_k &= -k_1 \spowf{x_k + T u_k}{1}{2} + v_{k+1} \\
        v_{k+1} &\in v_k -k_2 T \spow{x_k + T u_k}{0}.
    \end{align}
\end{subequations}
\changed{However, they perform the discretization considering only the unperturbed case.
As a consequence, this controller---unlike the continuous-time super-twisting controller---is not capable of rejecting unbounded disturbances with bounded slope.}
To see this, note that with abbreviations $w_{-2} := w_0$, $w_{-1} := w_0$ the trajectory
$
    x_k = T w_{k-1},
u_k = -w_{k-1},
v_k = -w_{k-2}
    $
is always a solution of the closed loop \eqref{eq:plant:discrete}, \eqref{eq:impl0}, provided that $k_2 > L$.
As a consequence, \changed{the controller \eqref{eq:impl0}} can \changed{only} guarantee that $|x_k| \le WT$ holds after a finite number of steps.

\subsection{Semi-Implicit Conditioned Super-Twisting Control}

Consider now the case of actuator saturation, i.e., the case that the control input $u$ is bounded by some control input bound $U \in \RR_{> 0}$ with $U > W$.
In this case, the classical super-twisting controller \eqref{eq:sta} may suffer from a windup effect that deteriorates control performance.

A continuous-time control law that mitigates this windup effect is the conditioned super-twisting controller proposed by \cite{seeber2020conditioned}.
Its control law is obtained by applying the conditioning technique by \cite{hanus1987conditioning} to \eqref{eq:sta} and is given by
\begin{subequations}
    \label{eq:sta:conditioned}
    \begin{align}
        \bar u &= -k_1 \spowf{x}{1}{2} + v \\
        u &= \sat_{U}(\bar u) = \begin{cases}
            \bar u & |\bar u| \le \umax \\
            U \sign(\bar u) & |\bar u| > \umax
        \end{cases}\\
    \label{eq:sta:conditioned:v}
        \dot v &= -k_2 \sign(v - u).
    \end{align}
\end{subequations}
\cite{yang2022semi} propose a semi-implicit discretization of this controller.
However, as shown in \citep{seeber2023discussion}, that discretization may suffer from limit cycles which deteriorate the achievable performance compared to the unsaturated case.

The present paper proposes new, proper implicit discretizations of both, the super-twisting controller and the conditioned super-twisting controller, such that the best possible worst-case performance shown in Proposition~\ref{prop:best} is achieved in either case.

\section{Implicit Super-Twisting Control \\ without Actuator Saturation}
\label{sec:implicit}

In the following, a new implicit discretization of the super-twisting controller with best possible worst-case disturbance rejection is derived.
First, note that \eqref{eq:impl0} drives the modified sliding variable $\tilde x_k = x_k - T w_{k-1}$ to zero, \changed{i.e., the variable $\tilde x_k$ defines the discrete-time sliding mode}, \changedc{cf. \cite{acary2012chattering}}. This variable satisfies
\begin{subequations}
    \begin{align}
        \label{eq:xtildekp1:1}
        \tilde x_{k+1} &= x_k + T u_k \\
        \label{eq:xtildekp1:2}
        &= \tilde x_{k} + T (u_k + w_{k-1}).
    \end{align}
\end{subequations}
In \citep{brogliato2020implicit}, this variable was chosen such that \eqref{eq:xtildekp1:1} does not depend on $w_k$, because otherwise the unknown quantity $w_k$ would be needed to compute $u_k$.

From \eqref{eq:xtildekp1:2}, one may see that $v_{k+1}$ in \eqref{eq:impl0:u} eventually compensates for $w_{k-1}$.
Using this intuition and the fact that the difference of two successive disturbance values satisfies $|w_{k-1} - w_{k-2}|\le LT$, an alternative modified sliding variable is proposed as
\begin{align}
    \label{eq:zk}
    z_k &= x_k - T(w_{k-2} + v_k) - T(w_{k-1} - w_{k-2}) \nonumber \\
    &= x_k - T (w_{k-1} + v_k).
\end{align}
If $z_k$ and $v_k+w_{k-2}$ are driven to zero, then $x_k$ satisfies the desired bound $|x_k| \le T |w_{k-1} - w_{k-2}| \le LT^2$.
By combining~\eqref{eq:plant:discrete} and~\eqref{eq:zk} to
\begin{subequations}
    \begin{align}
        \label{eq:zkp1:1}
        z_{k+1} &= x_k + T(u_k - v_{k+1})  \\
        \label{eq:zkp1:2}
        &= z_k + T (u_k + w_{k-1} + v_k - v_{k+1}),
\end{align}
\end{subequations}
one can see that the proposed modified sliding variable $z_k$ also has the property that the prediction \eqref{eq:zkp1:1} does not depend on the unknown quantity $w_k$.
The controller~\eqref{eq:impl0} contains the term $v_{k+1}$ to compensate for $w_{k-1}$ in ~\eqref{eq:xtildekp1:2}. This suggests that the control law that drives $z_k$ in~\eqref{eq:zkp1:2} to zero should contain the term $2 v_{k+1} - v_k$ to compensate for $w_{k-1} + v_k - v_{k+1}$.
The proposed proper implicit discretization of the super-twisting controller without actuator saturation is hence given by
\begin{subequations}
    \label{eq:impl3}
    \begin{align}
        \label{eq:impl3:u}
        u_{k} &= -k_1 \spowf{z_{k+1}}{1}{2} + 2 v_{k+1} - v_k \\
        \label{eq:impl3:v}
        v_{k+1} &\in v_k - T k_2  \spow{z_{k+1}}{0}.
    \end{align}
\end{subequations}
The next two theorems show how to implement this control law in explicit form and give closed-loop stability conditions.
Their proofs are given in Sections~\ref{sec:explicit:unsaturated} and~\ref{sec:stability:unsaturated}.
\begin{thm}
    \label{th:explicit}
    Let $k_1, k_2, T \in \RR_{> 0}$ and define the abbreviation $\lambda = k_2 - \frac{k_1^2}{4}$.
    Then, the explicit form of the implicit control law \eqref{eq:impl3}, i.e., the unique solution $u_k, v_{k+1}$ of the system of generalized equations \eqref{eq:zkp1:1} and \eqref{eq:impl3}, is given by
    \begin{subequations}
            \label{eq:impl3:explicit}
        \begin{align}
            \label{eq:impl3:explicit:u}
            &u_k = \begin{cases}
v_k -  \left(2 \lambda T + k_1 \sqrt{|x_k| - \lambda T^2} \right ) \sign(x_k) & \frac{|x_k|}{T^2} > k_2 \\
                v_k - \frac{2 x_k}{T} & \frac{|x_k|}{T^2} \le k_2
            \end{cases} \\
            \label{eq:impl3:explicit:v}
            &v_{k+1} = \begin{cases}
                v_k - T k_2 \sign(x_k) & \frac{|x_k|}{T^2} > k_2 \\
                v_k - \frac{x_k}{T} & \frac{|x_k|}{T^2} \le k_2
            \end{cases}
        \end{align}
    \end{subequations}
    for every given $x_k, v_k \in \RR$.
\end{thm}
\begin{rem}
It is worth noting that the proposed explicit control law~\eqref{eq:impl3:explicit} coincides with the discrete-time controller proposed by~\cite{andritsch2024modified}. However, the presented derivation provides an intuitive motivation that allows for a complete stability proof.
\end{rem}
\begin{rem}
    It is remarkable that for $k_1 = 2 \sqrt{k_2}$, the proposed \emph{implicit} discretization of the super-twisting controller has the particularly simple form
    \begin{equation}
            u_k = - k_1 \spowf{x_k}{1}{2} + v_k, \quad
            v_{k+1} = v_k - T k_2 \sign(x_k)
\end{equation}
    for $|x_k| > k_2 T^2$, which corresponds to the super-twisting controller with \emph{explicit} Euler discretization.
    Also, it becomes a second-order dead-beat controller for ${|x_k| \le k_2 T^2}$ regardless of $k_1$.
\end{rem}
\begin{thm}
    \label{th:stability}
    Let $L \in \RR_{\ge 0}$, $T \in \RR_{> 0}$ and consider the interconnection of the control law \eqref{eq:impl3:explicit}, \changed{the zero-order hold \eqref{eq:zoh}, and the continuous-time plant \eqref{eq:plant}.}
    Suppose that the disturbance \changed{$w : \RR_{\ge 0} \to \RR$ is Lipschitz continuous, fulfilling $|\dot w(t)| \le L$ almost everywhere,} and that the controller parameters $k_1, k_2 \in \RR_{> 0}$ satisfy 
    \begin{equation}
        k_1 > \sqrt{k_2 + L}, \qquad
        k_2 > L.
    \end{equation}
    \changed{Define $x_k = x(kT)$ and $w_k$ as in \eqref{eq:wk}.
    Then, an integer $K$ exists such that $v_k = - w_{k-2}$, $x_k = T(w_{k-1} - w_{k-2})$ hold for all $k\ge K$, and $|x(t)| \le LT^2$ holds for all $t \ge KT$.}
\end{thm}

\section{Implicit Conditioned Super-Twisting Control with Actuator Saturation}
\label{sec:conditioned}

In order to obtain an implicit discretization of the conditioned super-twisting controller, note that, in continuous time, its construction \changedc{in \cite{seeber2020conditioned}} is based on the fact that, in \eqref{eq:sta:conditioned:v},
\begin{equation}
    \label{eq:signequivalence}
    \sign(v-u) = \sign(k_1 \spowf{x}{1}{2}) = \sign(x)
\end{equation}
holds for $k_1 > 0$ in the unsaturated case $u = \bar u$.
Applying a similar modification to \eqref{eq:impl3} yields the proposed implicit conditioned super-twisting controller as
\begin{subequations}
    \label{eq:implc}
    \begin{align}
\label{eq:implc:u}
        \bar u_{k} &= -k_1 \spowf{z_{k+1}}{1}{2} + 2 v_{k+1} - v_k \\
        \label{eq:implc:usat}
        u_k &= \sat_{\umax} (\bar u_k) \\
    \label{eq:implc:v}
    v_{k+1} &\in v_k - T k_2  \spow{2 v_{k+1} - v_k - u_k}{0}.
    \end{align}
\end{subequations}
The next two theorems show how to implement this control law in explicit form and give closed-loop stability conditions.
Their proofs are given in Sections~\ref{sec:explicit:saturated} and~\ref{sec:stability:saturated}.
\begin{thm}
    \label{th:explicit:conditioned}
    Let $k_1, k_2, T, U \in \RR_{> 0}$ and define the abbreviation $\lambda = k_2 - \frac{k_1^2}{4}$.
    Then, an explicit form of the implicit conditioned super-twisting controller \eqref{eq:implc}, i.e., a solution $u_k, v_{k+1}$ of the system of generalized equations \eqref{eq:zkp1:1} and \eqref{eq:implc}, is given by
    \begin{subequations}
        \label{eq:implc:explicit}
        \begin{align}
            \label{eq:implc:explicit:u}
            &\hat u_k = \begin{cases}
v_k -  \left(2 \lambda T + k_1 \sqrt{|x_k| - \lambda T^2} \right ) \sign(x_k) & \frac{|x_k|}{T^2} > k_2 \\
                v_k - \frac{2 x_k}{T} & \frac{|x_k|}{T^2} \le k_2
            \end{cases} \\
            \label{eq:implc:explicit:usat}
            &u_k = \sat_{\umax} (\hat u_k) \\
            \label{eq:implc:explicit:v}
            &v_{k+1} = \begin{cases}
                v_{k} - T k_2 \sign(v_k - u_k) & |v_k - u_k| > 2 k_2 T \\
                \frac{v_k + u_k}{2} & |v_k - u_k| \le 2 k_2 T
            \end{cases}
\end{align}
    \end{subequations}
    for every given $x_k, v_k \in \RR$.
\end{thm}
\changed{\begin{rem}
    \mbox{Formally setting $U = \infty$ in \eqref{eq:implc:explicit} yields \eqref{eq:impl3:explicit}.}
\end{rem}}
\begin{rem}
    Note that the auxiliary variable $\hat u_k$ in the explicit form \eqref{eq:implc:explicit} is not necessarily equal to the auxiliary variable $\bar u_k$ in the implicit form \eqref{eq:implc}  \changedb{when $\abs{u_k} = U$}.
\end{rem}
\begin{thm}
    \label{th:stability:conditioned}
    Let $L, \wmax \in \RR_{\ge 0}$, $T \in \RR_{> 0}$.
    Consider the interconnection of the control law \eqref{eq:implc:explicit}, \changed{the zero-order hold \eqref{eq:zoh}, and the continuous-time plant \eqref{eq:plant} with a bounded, Lipschitz continuous disturbance $w : \RR_{\ge 0} \to \RR$ satisfying $|\dot w(t)| \le L$ and $|w(t)| \le \wmax$ almost everywhere.}
    Suppose that the control input bound \changed{$U \in \RR_{> 0}$ and the parameters $k_1, k_2 \in \RR_{> 0}$ satisfy} $\umax > \wmax + k_2 T$ and
    \begin{equation}
        k_1 > \sqrt{2 k_2 \frac{\umax + \wmax}{\umax - \wmax - k_2 T}}, \qquad
        k_2 > \wpmax.
    \end{equation}
    \changed{Define $x_k = x(kT)$ and $w_k$ as in \eqref{eq:wk}.
    Then, an integer $K$ exists such that $v_k = - w_{k-2}$, $x_k = T(w_{k-1} - w_{k-2})$ hold for all $k\ge K$, and $|x(t)| \le LT^2$ holds for all $t \ge KT$.}
\end{thm}

\changedc{It is perhaps interesting to compare the implicit conditioned super-twisting controller to the implicit first-order sliding mode controller $u_k \in -c \spow{x_k + T u_k}{0}$ studied in \cite{acary2012chattering}, which may equivalently be written in explicit form using a saturation function\footnote{\changedc{It is worth noting that \eqref{eq:impl3:explicit} and \eqref{eq:implc:explicit} also exhibit some kind of saturation-like form; in particular, \eqref{eq:impl3:explicit:v} and \eqref{eq:implc:explicit:v} may equivalently be written as $v_{k+1} = v_k - \sat_{k_2 T}(\frac{x_k}{T})$ and $v_{k+1} = v_k - \sat_{k_2 T}(\frac{v_k - u_k}{2})$, respectively.}} as $u_k = -\sat_c(\frac{x_k}{T})$, cf. \cite{brogliato2023comments}.
    Provided that the parameters of the former fulfill Theorem~\ref{th:stability:conditioned} and the parameter of the latter satisfies $c \in (W,U]$, both controllers comply with the control input bound $|u_k| \le U$ and, in the disturbance free case $w(t) \equiv 0$, also achieve finite-time stabilization of the origin without chattering.
    They differ in terms of disturbance rejection, however: the conditioned super-twisting controller has $|x_k| \le L T^2$ according to Theorem~\ref{th:stability:conditioned}, while the first-order sliding mode controller achieves $|x_k| \le W T$ according to \cite[Proposition~1]{acary2012chattering}.
Thus, the former has the capacity for significantly better accuracy if the largest change of the disturbance during a sampling step, i.e. $LT$, is significantly smaller than\footnote{\changedc{Note that, even in case $LT > W$, the accuracy of the implicit conditioned super-twisting controller is never worse than twice the first-order sliding mode band, because Theorem~\ref{th:stability:conditioned} additionally guarantees $|x_k| \le 2W T$ due to the trivial bound $|w_{k-1} - w_{k-2}| \le 2W$.}} the disturbance amplitude $W$.
Indeed, the implicit conditioned super-twisting---being an integral-type controller---can fully reject nonzero disturbances that are constant (i.e., $L = 0$), while the implicit first-order sliding mode controller exhibits a residual control error in such cases.}

\section{Derivation of Explicit Control Laws}
\label{sec:explicit}

This section formally derives the explicit forms of the control laws in Sections~\ref{sec:implicit} and~\ref{sec:conditioned}.

\subsection{Unsaturated Control Input}
\label{sec:explicit:unsaturated}

The explicit control law \eqref{eq:impl3:explicit} of the proposed implicit super-twisting controller is obtained in a similar fashion as in \citep{brogliato2020implicit}, using the following auxiliary lemma.
Its proof is given in \changed{Appendix~\ref{proof-lem:zkp1:explicit}}.
\begin{lem}
    \label{lem:zkp1:explicit}
    Let $k_1, k_2, T \in \RR_{> 0}$, define $\lambda = k_2 - \frac{k_1^2}{4}$, and let $x_k, v_k \in \RR$.
    Then, the unique solution $z_{k+1}$ of \eqref{eq:zkp1:1}, \eqref{eq:impl3} is given by
    \begin{equation}
        \label{eq:impl3:explicit:z}
        z_{k+1} = \begin{cases}
            \left(\sqrt{|x_k| - \lambda T^2} - \frac{T k_1}{2}\right)^2 \sign(x_k) & |x_k| > k_2 T^2 \\
            0 & |x_k| \le k_2 T^2.
        \end{cases}
    \end{equation}
\end{lem}
\changed{\begin{rem}
        Alternatively, the framework of monotone operators and their resolvents~\citep[cf.][Chapter 23]{bauschke2011} could be used to solve for $z_{k+1}$ in \eqref{eq:zkp1:1}, \eqref{eq:impl3}.
        In this case, the resolvent has the same structure as obtained in~\citep{mojallizadeh} for the implicit super-twisting differentiator.
\end{rem}}\begin{proof}[Proof of Theorem~\ref{th:explicit}]
    Using the unique solution $z_{k+1}$ from Lemma~\ref{lem:zkp1:explicit}, distinguish the two cases in \eqref{eq:impl3:explicit:z}, \changedc{i.e., $|x_k| > k_2 T^2$ and $|x_k| \le k_2T^2$}.
    In the first case, $\sign(z_{k+1}) = \sign(x_k)$ and \eqref{eq:impl3:v} yield \eqref{eq:impl3:explicit:v}, and from \eqref{eq:impl3:u} one obtains
\begin{align}
    u_k - v_k&= -k_1 \spowf{z_{k+1}}{1}{2}  + 2 v_{k+1} - 2 v_k \nonumber \\
    &= \left( \frac{T k_1^2}{2} -k_1 \sqrt{|x_k| - \lambda T^2} - 2 T k_2\right) \sign(x_k) \nonumber \\
    &= - \left( 2\lambda T + k_1 \sqrt{|x_k|-\lambda T^2} \right) \sign(x_k).
\end{align}
In the second case, $0 = z_{k+1} = x_k + T(v_{k+1} - v_k)$ yields \eqref{eq:impl3:explicit:v}, and \eqref{eq:impl3:u} yields $u_k = 2 v_{k+1} - v_k = v_{k} - \frac{2x_k}{T}$.
\end{proof}

\subsection{Saturated Control Input}
\label{sec:explicit:saturated}

Obtaining the explicit form of the implicit conditioned super-twisting controller \eqref{eq:implc} requires solving the system of generalized equations \eqref{eq:zkp1:1}, \eqref{eq:implc} containing the nonlinear saturation function.
The following lemma reduces this problem to the solution of the unsaturated equations \eqref{eq:zkp1:1}, \eqref{eq:impl3} \changed{with variables $z_{k+1}, v_{k+1}, u_k$ renamed to $\hat z_{k+1}, \hat v_{k+1}, \hat u_{k}$}.
The proof is in \changed{Appendix~\ref{proof-lem:implc:equivalence}}.

\begin{lem}
    \label{lem:implc:equivalence}
    Let $k_1, k_2, T, U \in \RR_{> 0}$ and $x_k, v_k \in \RR$.
    Consider the unique solution $\hat u_k$ of the system of generalized equations
    \begin{subequations}
        \label{eq:impl:aux}
        \begin{align}
            \label{eq:impl:aux:z}
            \hat z_{k+1} &= x_k + T (\hat u_k - \hat v_{k+1}) \\
            \label{eq:impl:aux:u}
            \hat u_k &= -k_1 \spowf{\hat z_{k+1}}{1}{2} + 2 \hat v_{k+1} - v_k \\
            \label{eq:impl:aux:v}
\changed{\hat v_{k+1}} &\changed{\in v_k - T k_2 \spow{\hat z_{k+1}}{0}.}
        \end{align}
    \end{subequations}
    Then, $u_k = \sat_U(\hat u_k)$ is a solution of the generalized system of equations \eqref{eq:zkp1:1}, \eqref{eq:implc}.
\end{lem}

\begin{proof}[Proof of Theorem~\ref{th:explicit:conditioned}]
\changed{Apply Theorem~\ref{th:explicit} to system \eqref{eq:impl:aux} to see from \eqref{eq:impl3:explicit:u} that \eqref{eq:implc:explicit:u} is its unique solution.}
    Lemma~\ref{lem:implc:equivalence} then yields \eqref{eq:implc:explicit:usat}.
\changed{To show, finally,} that \eqref{eq:implc:explicit:v} is the unique solution of \eqref{eq:implc:v}, define $a_{k+1} = 2 v_{k+1} - v_k - u_k$, $b_k = v_k - u_k$ and rewrite \eqref{eq:implc:v} as
    $
        a_{k+1} \in b_k - 2Tk_2 \spow{a_{k+1}}{0}.
        $
        Its unique solution is $a_{k+1} = b_k - 2 T k_2 \sign(b_k)$ for $|b_k| > 2 k_2 T$ and $a_{k+1} = 0$ otherwise, from which \eqref{eq:implc:explicit:v} follows.
\end{proof}

\section{Stability Analysis}
\label{sec:stability}

The stability analysis is performed by proving forward invariance and finite-time attractivity of certain sets according to the following definition.
\begin{defn}
    Consider trajectories of a discrete-time system, i.e., sequences $(\x_k)$ with $\x_k \in \RR^n$.
    A set $\Omega \subseteq \RR^n$ is called
    \begin{itemize}
        \item 
            \emph{forward invariant} along the trajectories, if for all trajectories $(\x_k)$, $\x_k \in \Omega$ implies $\x_{k+1} \in \Omega$ for all $k \in \NN_0$
        \item
            \emph{finite-time attractive} along the trajectories, if for each trajectory $(\x_k)$, there exists $K \in \NN_0$ depending only on $\x_0$ such that $\x_k \in \Omega$ holds for all $k \ge K$.
    \end{itemize}
\end{defn}

\subsection{Unsaturated Control Input}
\label{sec:stability:unsaturated}

Consider the closed loop formed by interconnecting the plant \eqref{eq:plant:discrete} without actuator constraints and the proposed control law \eqref{eq:impl3}.
To investigate its stability properties, consider the state variables $z_k$ and $q_k$ defined as, cf. \eqref{eq:zk},
\begin{equation}
    \label{eq:qk}
    z_k = x_k - T(w_{k-1} + v_k), \quad
    q_k = v_k + w_{k-2},
\end{equation}
with the definition $w_{-k} := w_0$ for $k \in \NN$.
According to \eqref{eq:zkp1:2} and \eqref{eq:impl3}, these are governed by
\begin{subequations}
    \label{eq:closedloop:discrete}
    \begin{align}
        z_{k+1} &= z_k - T k_1 \spowf{z_{k+1}}{1}{2} + T q_{k+1} \\
        \label{eq:closedloop:discrete:q}
        q_{k+1} &\in q_k - T k_2 \spow{z_{k+1}}{0} + T \delta_{k+1}
    \end{align}
    with the abbreviation
    \begin{equation}
        \delta_{k} = \frac{w_{k-2} - w_{k-3}}{T} \quad\text{satisfying $|\delta_k| \le L$.}
    \end{equation}
\end{subequations}
Similar to \citep{brogliato2020implicit}, the stability analysis is based on the fact that \eqref{eq:closedloop:discrete} may be interpreted as the implicit discretization of the continuous-time closed-loop system, \changed{understood in the sense of \cite{filippov1988differential},}\begin{subequations}
    \label{eq:closedloop:continuous}
    \begin{align}
        \dot z &= -k_1 \spowf{z}{1}{2} + q \\
        \dot q &= -k_2 \sign(z) + \delta, \qquad \abs{\delta} \le \wpmax
    \end{align}
\end{subequations}
obtained by applying the continuous-time super-twisting controller \eqref{eq:sta} to \eqref{eq:plant} with $z = x$ and $q = v + w$.

Stability properties of the discrete-time closed loop may hence be analyzed using a Lyapunov function \changedb{that is quasiconvex, i.e., that has} convex sublevel sets.
\changedb{The next lemma, proven in Appendix~\ref{proof-lem:lyap:discrete}, generalizes \cite[Lemma~5]{brogliato2020implicit} to quasiconvex Lyapunov functions which are only locally Lipschitz continuous and may hence be analyzed using Clarke's generalized gradient, cf. e.g. \cite[Section~5.4]{polyakov2014stability}.}
\changedb{\begin{lem}
        \label{lem:lyap:discrete}
        Let $\F : \RR^n \to 2^{\RR^n}$ be upper semicontinuous and $\F(\x)$ be nonempty and compact for all $\x \in \RR^n$.
        Let $V : \RR^n \to \RR_{\ge 0}$ be continuous, quasiconvex, positive definite, and locally Lipschitz continuous on $\RR^n \setminus \{ \bm{0} \}$.
        Denote by $\partial V : \RR^n \setminus \{\bm{0} \} \to 2^{\RR^n}$ its Clarke generalized gradient.
        Suppose that, for each $\x \in \RR^n \setminus \{ \bm{0} \}$,
        \begin{equation}
            \label{eq:lyap:condition}
            \max_{\substack{\h \in \F(\x) \\\vzeta \in \partial V(\x)}} \vzeta\TT \h < 0
        \end{equation}
        holds.
        Then, for each $T \in \RR_{> 0}$ there exists a negative definite, upper semicontinuous
function $Q_T : \RR^n \to \RR_{\ge 0}$ such that
        $
            V(\x_{k+1}) - V(\x_k) \le Q_T(\x_{k+1})
            $
        holds for all solutions of the inclusion $\x_{k+1} \in \x_k + T \F(\x_{k+1})$.
\end{lem}
\begin{rem}
    Condition~\eqref{eq:lyap:condition} essentially means that $\dot V$ is negative along trajectories of the system $\dot \x \in \F(\x)$, i.e., that $V$ is a strict Lyapunov function for that system.
\end{rem}}

\changedb{The} Lyapunov function from \citep{seeber2017stability} \changedb{is now used}; \changedb{it is shown to be quasiconvex} in the following lemma, which is proven in  \changed{Appendix~\ref{proof-lem:lyap}}.
\begin{lem}
    \label{lem:lyap}
    Let $L \in \RR_{\ge 0}$, $k_1, k_2 \in \RR_{> 0}$ and consider the function $V_\alpha : \RR^2 \to \RR$ defined as
    \begin{equation}
        \label{eq:V}
        V_{\alpha}(z,q) = \begin{cases}
            2 \sqrt{q^2 + 3 \alpha^2 k_1^2 z} - q &  z > 0, q < \alpha k_1 \spowf{z}{1}{2} \\
            2 \sqrt{q^2 - 3 \alpha^2 k_1^2 z} + q &  z < 0, q > \alpha k_1 \spowf{z}{1}{2} \\
            3 |q| & \text{otherwise}
        \end{cases}
    \end{equation}
    with $\alpha \in (0,1)$.
    Then, $V_\alpha$ is continuous and positive definite, \changedb{locally Lipschitz continuous except in the origin}, and its sublevel sets $\Omega_{\alpha,c} = \{ (z,q) \in \RR^2 : V_\alpha(z,q) \le c \}$ are convex for all $c \ge 0$\changedb{, i.e., it is quasiconvex}.
    Moreover, if $k_1 > \frac{1}{\alpha} \sqrt{k_2 + L}$, $k_2 > L$, then $V_\alpha$ is a strict Lyapunov function for the continuous-time closed loop \eqref{eq:closedloop:continuous}, \changedb{i.e., it satisfies \eqref{eq:lyap:condition} for the corresponding Filippov inclusion}.
\end{lem}

Using this Lyapunov function $V_\alpha$ \changedb{and Lemma~\ref{lem:lyap:discrete}}, the following lemma shows forward invariance and finite-time attractivity of its sublevel sets $\Omega_{\alpha,c}$.
Moreover, the origin is shown to be finite-time attractive by virtue of another forward invariant set $\Omega$.
The proof is given in  \changed{Appendix~\ref{proof-lem:invariance}}.

\begin{lem}
    \label{lem:invariance}
    Let $k_1, k_2, T \in \RR_{> 0}$, $L \in \RR_{\ge 0}$, and let the function $V_\alpha$ be defined as in Lemma~\ref{lem:lyap} with $\alpha \in (0,1)$.
    Suppose that $k_2 > L$.
    Consider the closed loop formed by the interconnection of \eqref{eq:plant:discrete} and \eqref{eq:impl3} with $w_k$ satisfying $|w_{k} - w_{k-1}| \le LT$, and consider the trajectories $(z_k, q_k)$ of $z_k$ and $q_k$ defined in \eqref{eq:qk}.
    Then, the following sets are forward invariant and finite-time attractive along closed-loop trajectories:
    \begin{enumerate}[(a)]
        \item
            \label{it:Omegac}
            $\Omega_{\alpha,c} = \{ (z,q) \in \RR^2 : V_{\alpha}(z,q) \le c \}$ for all $c \in \RR_{> 0}$, if $k_1 > \frac{\sqrt{k_2 + L}}{\alpha}$,
        \item
            \label{it:Omega}
            $\Omega = \{ (z,q) \in \RR^2 : \max\{ |z|, |z+Tq| \} \le (k_2-L)T^2 \}$, if $k_1 > \sqrt{k_2 + L}$.
\end{enumerate}
    Moreover, $(z_k, q_k) \in \Omega$ implies $z_{k+2} = q_{k+2} = 0$ for all $k \in \NN_0$.
\end{lem}
\begin{rem}
    Item~(\ref{it:Omegac}) of this lemma implies asymptotic stability of the origin of \eqref{eq:closedloop:discrete}, and item~(\ref{it:Omega}) implies its finite-time attractivity, for all admissible disturbances.
\end{rem}

\begin{proof}[Proof of Theorem~\ref{th:stability}]
    From Lemma~\ref{lem:invariance}, item~(\ref{it:Omega}), there exists $\tilde K \in \NN_0$ such that $(z_k, q_k) \in \Omega$, and consequently $z_{k+2} = q_{k+2} = 0$ for all $k \ge \tilde K$.
    Thus, $v_k - w_{k-2} = 0$ \changed{and $x_k = z_k + T q_k + T (w_{k-1} - w_{k-2}) = T(w_{k-1} - w_{k-2})$ are}
obtained from \eqref{eq:zk}, \eqref{eq:qk} for all $k \ge K = \tilde K +2$.
\changed{Noting that $|\dot w(t)| \le L$ implies $|w_{k+1} - w_k| \le LT$, the bound $|x_k| \le LT^2$ for all integers $k \ge K$ follows.}

\changed{To prove $|x(t)| \le LT^2$ for all $t \in [KT, \infty)$, suppose to the contrary---without restricting generality---that $k \ge K$, $t \in (kT, (k+1)T)$ exist with $x(t) > LT^2$.
    \changedc{Absolute} continuity \changedc{of $x$} and $x_k \le LT^2$ then guarantee existence of $\tau \in (kT, t)$ with $x(\tau) \ge LT^2$ and $0 < \dot x(\tau) = u_k + w(\tau)$.
    Now, modify the disturbance $w$ after $\tau$ such that it is kept constant on the interval $[\tau, (k+1)T]$.
After this modification, $|\dot w(t)| \le L$ still holds \changedc{almost everywhere} and $\dot x(t)$ is a positive constant on that interval, yielding the contradiction $x_{k+1} > x(\tau) \ge LT^2$.}
\end{proof}

\subsection{Saturated Control Input}
\label{sec:stability:saturated}

Consider now the closed loop formed by interconnecting the plant \eqref{eq:plant:discrete} and the proposed conditioned control law \eqref{eq:implc}.
In this case, it is more convenient to write the closed-loop dynamics using the variables $z_k$ and $v_k$ as well as the auxiliary unsaturated control input $\bar u_k$ as
\begin{subequations}
    \begin{align}
        z_{k+1} &= z_k + T \bigl(\sat_U(\bar u_k) - v_{k+1} + v_k + w_{k-1} \bigr) \\
        v_{k+1} &\in v_k - T k_2 \spow{2 v_{k+1} - v_k - \sat_U(\bar u_k)}{0} \\
        \bar u_k &= -k_1 \spowf{z_{k+1}}{1}{2} + 2 v_{k+1} - v_k.
    \end{align}
\end{subequations}
If the saturation is inactive, i.e., if $|\bar u_k|\le U$, then this closed loop reduces to the unsaturated closed loop and may be written as \eqref{eq:closedloop:discrete} with state variables $z_k$ and $q_k$.

The next lemma establishes forward invariance and global finite-time attractivity of a hierarchy of three sets $\mathcal{M}_1 \supset \mathcal{M}_2 \supset \mathcal{M}_3$.
It allows to conclude that $|\bar u_k|\le U$ is established and maintained indefinitely after a finite time as trajectories enter $\mathcal{M}_3$.
\begin{lem}
    \label{lem:invariance:saturated}
    Let $k_1, k_2, T, U, \delta \in \RR_{> 0}$ and $L,W  \in \RR_{\ge 0}$.
    Suppose that $U > W + k_2 T$ and $k_2 > L$.
    Consider the closed loop formed by the interconnection of \eqref{eq:plant:discrete} and \eqref{eq:implc}, with $w_k$ satisfying $|w_k| \le W$ and $|w_{k+1} - w_{k}| \le LT$ for all $k \in \NN_0$, and consider the trajectories $(z_k, v_k, \bar u_k)$ of $z_k$ defined in \eqref{eq:qk}, $v_k$ as in \eqref{eq:implc:v}, and $\bar u_k$ defined in \eqref{eq:implc:u}.
    Then, the following sets are forward invariant and finite-time attractive along closed-loop trajectories:
    \begin{enumerate}[(a)]
        \item
            \label{it:M1}
            $\mathcal{M}_1 = \{ (z,v,\bar u) \in \RR^3 : |v| \le U \}$,
        \item
            \label{it:M2}
            $\mathcal{M}_2 = \{ (z,v,\bar u) \in \mathcal{M}_1 : |z| \le \frac{(U+W+\delta)^2}{k_1^2} \}$,
        \item
            \label{it:M3}
            $\mathcal{M}_3 = \{ (z,v,\bar u) \in \mathcal{M}_2 : |\bar u| \le U \}$, if $k_1$ additionally satisfies
    \begin{align}
        \label{eq:cond:delta}
        k_1 > \sqrt{2 k_2 \frac{\umax + \wmax + \delta}{\umax - \wmax - k_2 T}}.
    \end{align}
    \end{enumerate}
\end{lem}
\changedb{The proof of the lemma is given in  Appendix~\ref{proof-lem:invariance:saturated}.}
\begin{proof}[Proof of Theorem~\ref{th:stability:conditioned}]
    Choose $\delta \in \RR_{> 0}$ sufficiently small such that \eqref{eq:cond:delta} is satisfied.
    From Lemma~\ref{lem:invariance:saturated}, item~(\ref{it:M3}), there then exists $\tilde K_1 \in \NN_0$ such that $|\bar u_k| \le U$ holds for all $k \ge \tilde K_1$.
    Then, the saturation in \eqref{eq:implc:usat} becomes inactive, i.e., $u_k = \bar u_k$ holds for all $k \ge \tilde K_1$.
    Noting that $u_k$ then satisfies \eqref{eq:impl3:u}, that $v_{k+1}$ then fulfills
    \begin{equation}
        v_{k+1} \in v_k - T k_2 \spow{2 v_{k+1} - v_k - \bar u_k}{0} = v_k - T k_2 \spow{z_{k+1}}{0},
    \end{equation}
    i.e., \eqref{eq:impl3:v}, and that
    $
        k_1 > \sqrt{2 k_2} > \sqrt{k_2 + L}
        $
        holds, the claim then follows from Theorem~\ref{th:stability}.
\end{proof}

\section{Simulation Results}
\label{sec:simulation}

In order to demonstrate the effectiveness of the proposed discrete-time STC as well as the discrete-time conditioned STC, the following simulations were performed. In all simulations, the disturbance signal \mbox{$w(t) = W \eta(\frac{L}{W} (t - T) - 1)$} was applied, with the normalized sawtooth-function $\eta : \RR \to \RR$ defined as
\begin{align}\label{eq:normalized-sawtooth}
	\eta(t) &= |(t \text{ mod } 4) - 2| - 1 \nonumber \\
    &\changedc{= \begin{cases}
        \eta(t+4) & t < 0 \\
        1 - t & t \in [0, 2) \\
        t - 3 & t \in [2, 4) \\
        \eta(t-4) & t \ge 4,
\end{cases}}
\end{align}
$L=5$, $W=0.25$, and the sampling time $T=0.01$.
This disturbance $w(t)$ fulfills $|\dot w(t)| \leq L$ and $|w(t)| \leq W$, \changedc{and is displayed in Fig.~\ref{sim_disturbance}}. 
The corresponding discrete-time disturbance $w_k$ according to~\eqref{eq:wk} \changedc{is, in this case, given by
    \begin{equation}
        w_k = \begin{cases}
            0.05 k - 0.025 & k \in \{0, 1, 2, 3, 4, 5\} \\
            0.525 - 0.05 k & k \in \{6, 7, 8, 9 \} \\
            - w_{k-10} & k \ge 10,
        \end{cases}
    \end{equation}
and} is depicted as well.
Note that the discrete-time disturbance $w_k$ fulfills the corresponding discrete-time bounds $|w_{k+1}-w_k| \leq LT \changedc{= 0.05}$ and $|w_k| \leq W \changedc{= 0.25}$.

\begin{figure}
	\centering
    \includegraphics[width=\linewidth]{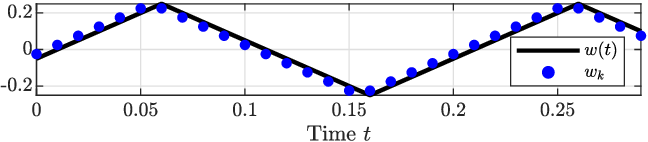}
    \caption{\changedc{Disturbance signal $w$ and corresponding discrete-time disturbance $w_k$ with $L = 5$, $W = 0.25$ and sampling time $T = 0.01$ applied in both simulations.}}\label{sim_disturbance}
\end{figure}

Fig.~\ref{sim_without_saturation} shows the results of the STC without actuator saturation. The proposed controller is compared with the semi-implicitly discretized STC by\changed{~\cite{xiong2022discrete}} and with the original implicit discretization by~\cite{brogliato2020implicit}. It can be observed that the original implicit discretization does not manage to drive the state $x_k$ into the best worst-case error band $|x_{k}| \leq LT^2$ from Proposition~\ref{prop:best}. Instead, the remaining control error is proportional to the disturbance $w_k$ itself. For the selected controller gains, $k_2 = 10$ and $k_1 = 27$, the semi-implicitly discretized STC shows a significantly larger convergence time compared to the implicit discretizations.

\begin{figure}
	\centering
    \includegraphics[width=\linewidth]{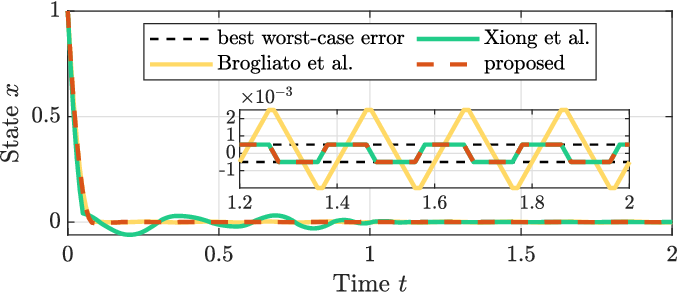}
	\caption{Results of the discrete-time STC without actuator saturation. Parameters: $k_2=10$, $k_1=27$.}\label{sim_without_saturation}
\end{figure}

\begin{figure}
	\centering
    \includegraphics[width=\linewidth]{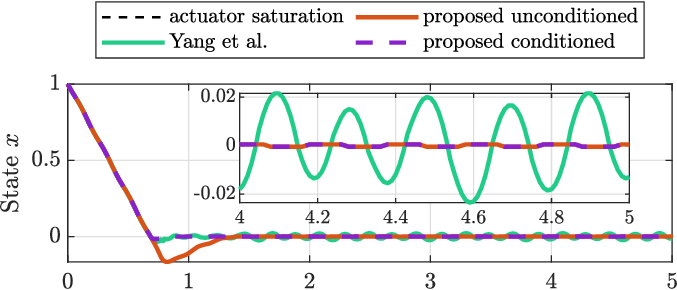}
    \includegraphics[width=\linewidth]{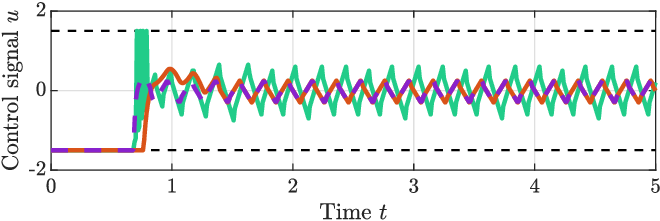}
	\caption{Results of the discrete-time STC in case of saturated actuation. Parameters: $k_2=10$, $k_1=16$, $U=1.5$. Top: plant state $x$, bottom: control signal $u$.}\label{sim_conditioned}
\end{figure}

Fig.~\ref{sim_conditioned} shows the results of the discrete-time STC in the case of an actuator saturation. The saturation was set to $U=1.5$. The proposed algorithms are compared to the conditioned STC by~\cite{yang2022semi}. The \changedc{proposed} conditioned STC stops the integration within the controller state when it is in the saturation, which the unconditioned controller does not. This leads to a reduced convergence time of the conditioned controller compared to the unconditioned controller and to a largely reduced undershoot of the conditioned controller compared to the unconditioned STC. The conditioned controller by~\cite{yang2022semi} also stops the integration of the controller state, which leads to a reduced convergence time as well compared to the unconditioned controller. However, for the selected parameters $k_2 = 10$ and $k_1 = 16$, the conditioned controller by~\cite{yang2022semi} does not yield the same accuracy as the proposed controllers. This result contradicts~\citep[Theorem 1]{yang2022semi} and was already addressed in~\citep{seeber2023discussion} in a counterexample. Also, upon the zero-crossing of the state $x$, the control signal of the conditioned controller by \cite{yang2022semi} exhibits a high-frequency switching behavior.

\section{Conclusion}
\label{sec:conclusion}

A new implicit discretization of the super-twisting controller was proposed.
In contrast to existing approaches, the proposed controller can handle the same class of disturbances as its continuous-time counterpart while also achieving best possible worst-case performance and being intuitive to tune.
For the case of constrained actuators, the proposed discretization was extended to the conditioned super-twisting controller.
The resulting implicit conditioned super-twisting controller mitigates windup by means of the conditioning technique and features similarly simple stability conditions as its continuous-time counterpart.
Numerical simulations demonstrated the superior performance of the proposed approach in comparison to existing approaches, as well as the proven stability and performance guarantees.
Future work may study extensions of the proposed discretization to other higher order sliding-mode control laws.

\section*{\changedc{Acknowledgement}}

\changedc{The authors would like to thank one anonymous reviewer for their extensive comments and suggestions, which have helped to improve the paper.}

\appendix

\section{Proofs}

The following auxiliary lemma is used in the proofs.
\begin{lem}
    \label{lem:auxlem}
    Let $Z \in \RR_{> 0}$ and  $z_k, z_{k+1} \in [-Z,Z]$.
    Suppose that $z_{k+1} \ge z_k$.
    Then,
    $
    \spowf{z_{k+1}}{1}{2}  \ge  \spowf{z_k}{1}{2} + \frac{z_{k+1} - z_k}{2 \sqrt{Z}}.
        $
\end{lem}
\begin{proof}
    Let $\alpha_k = z_{k+1} - z_k \ge 0$ and define the function $h : \RR_{\ge 0} \to \RR$ as
    $
    h(\alpha) = \spowf{z_{k} + \alpha}{1}{2} - \alpha/(2\sqrt{Z}).
        $
        Then, for all $\alpha \in [0, \alpha_k]$ its derivative fulfills $\deriv{h}{\alpha} \ge 0$, since \mbox{$|z_k+\alpha| \le Z$}.
Thus, $h(\alpha_k) \ge h(0)$.
\end{proof}

\subsection{Proof of Proposition~\ref{prop:best}}
\label{proof-prop-best}
\vspace{-0.4cm}
    For $M \in \NN$, define an auxiliary Lipschitz continuous function $\eta_M : \RR_{\ge 0} \to \RR$ as
    \begin{equation}
        \eta_M(t) = \begin{cases}
            \eta(t) & \text{if $t \in [0, 2M)$} \\
            (-1)^M [ 1 + (t-2M) ] & \text{if $t \in [2M, 2M+2)$} \\
            3 (-1)^M & \text{if $t \in [2M+2,\infty)$}
        \end{cases}
    \end{equation}
    with the sawtooth function $\eta$ defined as in~\eqref{eq:normalized-sawtooth}.
It is easy to verify that $\eta_M$ is Lipschitz continuous, and satisfies the inequalities $|\eta_M(t)| \le 3$ and $|\dot \eta_M(t)| \le 1$ almost everywhere.
    Moreover, $\int_{2\ell}^{2\ell+2} \eta(\tau)\,\diffd \tau = 0$ holds for all integers $\ell \in[0, M-1]$ and $\int_{2M}^{2M+2} \eta_M(\tau)\,\diffd \tau = (-1)^M \cdot 4$.

    Now, define $w(t) = -\frac{q LT}{2} \eta_{K+1}(\frac{2t}{T})$ with $q \in \{-1,1\}$ to be specified.
    Then, $w_k = 0$ for $k = 0, \ldots, K$ regardless of $q$, and $w_{K+1} = (-1)^{K} qLT$ according to \eqref{eq:wk}.
    Thus, $x_{0}, \ldots, x_{K+1}$ and hence also $u_{0}, \ldots, u_{K+1}$ are independent of $q$.
    Let
    \begin{equation}
        q = \begin{cases}
            (-1)^{K} & \text{if } 2 x_{K+1} - x_K + T (u_{K+1} - u_K) \ge 0 \\
            (-1)^{K+1} & \text{otherwise}.
        \end{cases}
    \end{equation}
    Then, \eqref{eq:plant:discrete} implies
    \begin{align}
        x_{K+2} &= 2 x_{K+1} - x_K + T ( u_{K+1} - u_K + w_{K+1} - w_K) \nonumber\\
        &= \frac{(-1)^{K}}{q} (|2 x_{K+1} - x_K + T ( u_{K+1} - u_K ) | + L T^2)\hspace{-0.3ex}
    \end{align}
    which yields $|x_{K+2}| \ge LT^2$, concluding the proof.
    \qed

\subsection{\changed{Proof of Lemma~\ref{lem:zkp1:explicit}}}\label{proof-lem:zkp1:explicit}
\vspace{-0.4cm}

Substituting \eqref{eq:impl3} into \eqref{eq:zkp1:1} yields the generalized equation
\begin{equation}
    \label{eq:generalizedeq}
    z_{k+1} \in x_k - T k_1 \spowf{z_{k+1}}{1}{2} - T^2 k_2 \spow{z_{k+1}}{0}.
\end{equation}
If $|x_k| \le k_2 T^2$, then its unique solution is $z_{k+1} = 0$.
Otherwise, $z_{k+1}$ and $x_k$ have the same sign, and multiplying \eqref{eq:generalizedeq} by $\sign(x_k) = \sign(z_{k+1})$ yields the equation
\begin{equation}
    |z_{k+1}| = |x_k| - T k_1 \apowf{z_{k+1}}{1}{2} - T^2 k_2
\end{equation}
whose unique solution is
\begin{align}
    \apowf{z_{k+1}}{1}{2} &= - \frac{T k_1}{2} + \sqrt{\frac{T^2 k_1^2}{4} - T^2 k_2 + |x_k|}.
\end{align}
Substituting $k_2 = \lambda + \frac{k_1^2}{4}$  yields \eqref{eq:impl3:explicit:z}.
\qed

\subsection{\changed{Proof of Lemma~\ref{lem:implc:equivalence}}}\label{proof-lem:implc:equivalence}
\vspace{-0.4cm}
    Distinguish cases $|\hat u_k| \le \umax$ and $|\hat u_k| > \umax$.
    In the first case, $\bar u_k = u_k = \hat u_k$, \changed{$z_{k+1} = \hat z_{k+1}$, $v_{k+1} = \hat v_{k+1}$ may be verified to be} a solution of \eqref{eq:zkp1:1}, \eqref{eq:implc} \changed{by using \eqref{eq:impl:aux}}.
    In the second case, suppose that $\hat u_k > \umax$; the proof for $\hat u_k < -\umax$ is obtained analogously.
    Set $u_k = U$, and let $v_{k+1}$, $\bar u_k$, and $z_{k+1}$ be uniquely defined by \eqref{eq:implc:v}, \eqref{eq:implc:u}, and \eqref{eq:zkp1:1}.
    It will be shown that $\bar u_k > U$, proving that also \eqref{eq:implc:usat} holds.
    To that end, distinguish the two cases $\hat v_{k+1} \le v_{k+1}$ and $\hat v_{k+1} > v_{k+1}$.
    In the first case,
    \begin{equation}
        \label{eq:zkp1:le:zhatkp1}
        z_{k+1} = x_k + T (u_k - v_{k+1}) \le x_k + T(\hat u_k - \hat v_{k+1}) = \hat z_{k+1}
    \end{equation}
    follows from \eqref{eq:zkp1:1}, \eqref{eq:impl:aux:z}, and thus \eqref{eq:implc:u}, \eqref{eq:impl:aux:u} yield
    \begin{align}
        \bar u_k &= -k_1 \spowf{z_{k+1}}{1}{2} + 2 v_{k+1} - v_{k} \nonumber \\
        &\ge -k_1 \spowf{\hat z_{k+1}}{1}{2} + 2 \hat v_{k+1} - v_{k}  = \hat u_k > U.
    \end{align}
\changed{For the second case, use \eqref{eq:impl:aux:u} and $\spow{k_1 \spowf{x}{1}{2}}{0} = \spow{x}{0}$ to rewrite \eqref{eq:impl:aux:v} as $\hat v_{k+1} \in v_k - T k_2 \spow{2 \hat v_{k+1} - v_k - \hat u_k}{0}$, and note that the expression $\spow{2 v_{k+1} - v_k - u_k}{0}$ in \eqref{eq:implc:v} exceeds the one in that inclusion due to $\hat v_{k+1} > v_k$; hence}
\begin{align}
            \label{eq:auximpl:ineq1}
\changed{2 v_{k+1} - v_k - u_k \ge 0 \ge 2 \hat v_{k+1} - v_k - \hat u_k} \end{align}
\changed{holds.
    Substituting \eqref{eq:impl:aux:u} yields $0 \ge \hat z_{k+1}$, and \eqref{eq:auximpl:ineq1} along with $\hat v_{k+1} > v_{k+1}$ further implies}
\begin{equation}
        u_k - v_{k+1} \le v_{k+1} - v_k < \hat v_{k+1} - v_k \le \hat u_k - \hat v_{k+1}.
    \end{equation}
Thus, $z_{k+1} < \hat z_{k+1} \le 0$ is concluded as in \eqref{eq:zkp1:le:zhatkp1}.
    Then,
    \begin{equation}
        \bar u_k >  2 v_{k+1} - v_{k} \ge u_k = U
    \end{equation}
    follows from \eqref{eq:implc:u}, \eqref{eq:auximpl:ineq1}, concluding the proof.
	\qed

    \subsection{\changedb{Proof of Lemma~\ref{lem:lyap:discrete}}}\label{proof-lem:lyap:discrete}
\vspace{-0.4cm}

\changedb{Define $Q_T(\x) = \max_{\h \in \F(\x)} V(\x) - V(\x - T \h)$, which is well-defined due to compactness of $\F(\x)$ and continuity of $V$.
To see upper semicontinuity, consider a sequence $(\x_i)$ with limit $\bar \x$ and corresponding $\h_i \in \F(\x_i)$ such that $Q_T(\x_i) = V(\x_i) - V(\x_i - T \h_i)$ and $\lim_{i\to \infty} Q_T(\x_i) = \limsup_{\x \to \bar\x} Q_T(\x)$.
        Then, $\h_i$ converges to the compact set $\F(\bar \x)$ due to upper semicontinuity of $\F$; thus, select subsequences such that $(\h_i)$ converges to some $\bar \h \in \F(\bar\x)$.
        Upper semicontinuity then follows from $\lim_{i\to\infty} Q_T(\x_i) = V(\bar \x) - V(\bar \x - T \bar \h) \le Q_T(\bar \x)$.

        To prove negative definiteness of $Q_T$, suppose to the contrary that there exist $\x \in \RR^n \setminus \{ \bm{0} \}$ and $\h \in \F(\x)$ such that $V(\x) - V(\x - T\h) \ge 0$.
        Since $V$ is locally Lipschitz at $\x$, $\partial V(\x)$ is nonempty and compact and $\partial V$ is upper semicontinuous at $\x$.
        Hence, $\vzeta \in \partial V(\x)$ exists such that $\vzeta\TT (- T\h) > 0$.
        Since $V$ is quasiconvex, application of \cite[Theorem 2.1]{daniilidis1999characterization} yields $V(\x - \lambda_1 T \h) \le V(\x - \lambda_2 T \h)$ for all $0 \le \lambda_1 \le \lambda_2 \le 1$.
        Consequently, $V(\x) = V(\x - \lambda T \h)$ for all $\lambda \in [0, 1]$, i.e., $V$ is constant on the line segment from $\x$ to $\x - T \h$.
        Thus, \cite[Lemma 2.1]{daniilidis1999characterization} yields $\vzeta\TT \h = 0$ for all $\vzeta \in \partial V(\x - \lambda T \h)$ and all $\lambda \in (0,1)$.
        Choose any sequence $(\lambda_i)$ tending to zero and converging $\vzeta_i \in \partial V(\x - \lambda_i T \h)$.
        Due to upper semicontinuity of $\partial V$, then $\bar \vzeta = \lim_{i \to \infty} \vzeta_i \in \partial V(\x)$, but $\bar \vzeta\TT \h = \lim_{i\to\infty} \vzeta_i\TT \h = 0$, contradicting \eqref{eq:lyap:condition}.
        \qed
}

\subsection{\changed{Proof of Lemma~\ref{lem:lyap}}}\label{proof-lem:lyap}
\vspace{-0.4cm}
Continuity and positive definiteness of $V_\alpha$ as well as the fact that it is a strict Lyapunov function\changedb{\footnote{\changedb{To verify condition \eqref{eq:lyap:condition} at points where $V_{\alpha}$ is not differentiable, note that Clarke's generalized gradient is the convex hull of adjacent (classical) gradients at such points.}}} for \eqref{eq:closedloop:continuous} under the stated conditions are shown in \cite[Section 3]{seeber2017stability}.
    \changedb{Local Lipschitz continuity outside the origin is obvious from the fact that the square root is zero only if $z = q = 0$.}
From the definition of $V_\alpha$ and its continuity, one can see that $(z, q) \in \Omega_{\alpha,c}$ if and only if the inequalities
            $|12 \alpha^2 k_1^2 z - 2 c q| \le c^2 - 3 q^2$,
$|q| \le \frac{c}{3}$
    hold \cite[see also][Fig.~1]{seeber2017stability}.
Since both inequalities are convex in $(z,q)$, the set $\Omega_{\alpha,c}$ is convex by virtue of being the intersection of two convex sets.
    \qed

\subsection{\changed{Proof of Lemma~\ref{lem:invariance}}}\label{proof-lem:invariance}
\vspace{-0.4cm}
\changedb{For item~(\ref{it:Omegac}), denote $\x = [z \quad q]\TT$ and define compact sets $\Lambda_b = \{ \x \in \RR^2 : V_\alpha(\x) \in [b, 2b]\}$.
        For each $b > 0$, existence of $\epsilon_b > 0$ will be shown such that $\x_{k+1} \in \Lambda_b$ implies $V_{\alpha}(\x_{k+1}) \le V_{\alpha}(\x_k) - \epsilon_b$, which implies finite-time attractivity and forward invariance of $\Omega_{\alpha,c}$.
        To that end, first relax \eqref{eq:closedloop:discrete} to $\x_{k+1} \in \x_k + T \F(\x_{k+1})$ with 
    \begin{equation}
        \F(z, q) = \begin{bmatrix} -k_1 \spowf{z}{1}{2} + q \\ -k_2 \spow{z}{0} + [-L,L] \end{bmatrix}.
    \end{equation}
    From Lemma~\ref{lem:lyap}, $V_{\alpha}$ is a strict Lyapunov function for $\dot \x \in \F(\x)$, i.e., condition \eqref{eq:lyap:condition} of Lemma~\ref{lem:lyap:discrete} is satisfied.
Since also the other conditions of the latter lemma are fulfilled, $V_\alpha(\x_{k+1}) - V_{\alpha}(\x_k) \le \max_{\x \in \Lambda_b} Q_T(\x) =-\epsilon_b$ holds whenever $\x_{k+1} \in \Lambda_b$; this maximum is well-defined due to upper semicontinuity of $Q_T$ and negative due to its negative definiteness.
    This proves item~(\ref{it:Omegac}).}

For item~(\ref{it:Omega}), choose $\alpha \in (0,1)$ sufficiently large such that $k_1 > \frac{\sqrt{k_2+L}}{\alpha}$.
Finite-time attractivity of $\Omega$ is then clear from the fact that it contains a finite-time attractive set $\Omega_{\alpha,c}$ with sufficiently small $c > 0$.
To show forward invariance of $\Omega$, it will be shown that \changed{$(z_{k+1},q_{k+1}) \notin \Omega$ implies $(z_k,q_k) \notin \Omega$.} 
Distinguish the cases $z_{k+1} \ne 0$ and $z_{k+1} = 0$.
In the first case, \changed{the contradiction}
\begin{align}
    &|z_k +T q_k| \nonumber \\
    &\quad= |z_{k+1} + Tk_1 \spowf{z_{k+1}}{1}{2} + T^2 k_2 \sign(z_{k+1}) - T^2 \delta_{k+1}| \nonumber \\
    &\quad> (k_2 - L) T^2.
\end{align}
\changed{is obtained by substituting $z_k$ and $q_k$ using \eqref{eq:closedloop:discrete}.}
In the second case, $|z_{k+1} + T q_{k+1}| > (k_2 - L)T^2$ implies
\begin{align}
    |z_k| &= |z_{k+1} - T q_{k+1} + T k_1 \spowf{z_{k+1}}{1}{2}| \nonumber \\
    &= |T q_{k+1}| > (k_2 - L)T^2.
\end{align}
Finally, $(z_k,q_k) \in \Omega$ implies $z_{k+2} = q_{k+2} = 0$, because then $(z_{k+1},q_{k+1}) \in \Omega$, yielding $z_{k+2} = 0$ as shown above, and thus $T q_{k+2} = z_{k+2} + T k_1 \spowf{z_{k+2}}{1}{2} - z_{k+1} = 0$.
\qed

\subsection{\changed{Proof of Lemma~\ref{lem:invariance:saturated}}}\label{proof-lem:invariance:saturated}
\vspace{-0.4cm}
\changed{For} item~(\ref{it:M1}), it is first shown that $\mathcal{M}_1$ is forward invariant.
    This is seen from the fact that $|v_{k+1}| > U$ and $|v_{k}| \le U$ imply the contradiction
    $
        |v_{k}| = |v_{k+1} + T k_2 \sign(v_{k+1})| > U
    $
    from \eqref{eq:implc:v}, because $\sign(2 v_{k+1} - v_k - u_k) = \sign(v_{k+1})$.
    To show finite-time attractivity, note that $v_k$ cannot change sign without entering $\mathcal{M}_1$, because $U > k_2 T$.
    Hence, without restriction of generality, it is sufficient to show that the assumption $v_k > U$ for all $k \in \NN_0$ leads to a contradiction.
    \changed{Under this assumption}, $v_k$ is strictly decreasing, because \changed{$v_{k+1} \ge v_k$} and \eqref{eq:implc:v} imply the contradiction $v_{k+1} \le v_k - k_2 T$.
    Since $v_k$ is also bounded from below, there exists $\kappa \in \NN$ such that $|v_{k+1} - v_k| \le \frac{k_2 T}{2}$ for all $k \ge \kappa$.
    \changed{Then, the right-hand side of \eqref{eq:implc:v} is truly multivalued, i.e., $2 v_{k+1} - v_k - u_k = 0$. Thus, $0 < U - \frac{k_2 T}{2} \le v_{k+1} + (v_{k+1}- v_k) = u_k \le \bar u_k$}
\changed{and, using \eqref{eq:zkp1:2}, $z_k$ is seen to strictly increase} according to
    \begin{align}
        z_{k+1} -z_k&= T (u_k + w_{k-1} - v_{k+1} + v_k) \\
        &= T (v_{k+1} + w_{k+1}) \ge T(U - W) \ge k_2 T^2, \nonumber
    \end{align}
    eventually leading to the contradiction \changed{$\bar u_k < 0$ in \eqref{eq:implc:u}} for sufficiently large $k > \kappa$, \changed{proving item~(\ref{it:M1})}.

    For item~(\ref{it:M2}), since $\mathcal{M}_2 \subset \mathcal{M}_1$ and due to item~(\ref{it:M1}), it is sufficient to consider trajectories in $\mathcal{M}_1$, i.e., to assume $|v_k| \le U$ \changed{for all $k$}.
    Let $\epsilon = \min\{ \delta, \umax - \wmax - k_2T \}$.
    It will be shown that $k_1^2 z_{k+1} > (\umax + \wmax + \delta)^2$ implies $z_k \ge z_{k+1} + \epsilon T$, from which the claim follows due to $\epsilon > 0$ and symmetry reasons.
    Distinguish the \changedc{three} cases $\bar u_k > \umax$, $\bar u_k < -\umax$, and $|\bar u_k| \le \umax$.
    The first case cannot occur, because, \changed{using \eqref{eq:implc:u} and $|v_{k+1} - v_k| \le k_2T$, the contradiction}
    \begin{align}
        2 W &+  k_2 T < U + W < k_1 \spowf{z_{k+1}}{1}{2} = 2 v_{k+1} - v_k - \bar u_k \nonumber \\
    &\quad< \changed{2 v_{k+1} - v_k - U \le } v_{k+1} - v_k \le k_2 T
    \end{align}
    \changed{is obtained.}
    In the second case, $u_k = -\umax$ and hence
    \begin{align}
        z_k &= z_{k+1} - T(u_k + w_{k-1} + v_k  -v_{k+1}) \nonumber \\
        &\ge z_{k+1} - T(- \umax + \wmax + k_2T) \ge z_{k+1} + T \epsilon.
    \end{align}
    is obtained from \eqref{eq:zkp1:2}.
    And in the third case,
    \begin{align}
        z_k &= z_{k+1} - T(\bar u_k + w_{k-1} + v_k  -v_{k+1}) \\
        &= z_{k+1} - T (-k_1 \spowf{z_{k+1}}{1}{2} + w_{k-1} + v_{k+1}) \nonumber\\
        &\ge z_{k+1} + T ( \umax + \wmax + \delta - \wmax - \umax) \ge z_{k+1} + T \epsilon, \nonumber
    \end{align}
    follows from $u_k = \bar u_k$ and \eqref{eq:implc:u}, proving item~(\ref{it:M2}).

    To show item~(\ref{it:M3}), since $\mathcal{M}_3 \subset \mathcal{M}_2$, it is again sufficient to consider trajectories in $\mathcal{M}_2$, i.e., to use the assumptions $k_1^2 |z_k| \le (\umax+\wmax+\delta)^2$ and $|v_k| \le U$ \changed{for all $k$}.
    Let $\epsilon = \frac{k_1^2}{2} \frac{\umax - \wmax - k_2 T}{\umax + \wmax + \delta} - k_2 > 0$.
    It will be shown that $\bar u_k > \umax$ implies $\bar u_{k-1} \ge \bar u_k + T \epsilon$; the claim then follows due to symmetry reasons.
    To see this, \changed{use $u_k \changed{= \sat_U(\bar u_k)} = \umax$ and \eqref{eq:zkp1:2} to obtain
    $
    z_{k+1} \ge z_{k} + T (U - W - k_2 T).
$}
        \changed{Then,} $\max\{ \apowf{z_k}{1}{2}, \apowf{z_{k+1}}{1}{2} \} \le \frac{\umax + \wmax + \delta}{k_1}$ and Lemma~\ref{lem:auxlem} \changed{imply}
    \begin{align}
        \changed{\spowf{z_{k+1}}{1}{2} }
&\changed{\ge
        \spowf{z_{k}}{1}{2} +} \frac{k_1 T}{2} \frac{\umax - \wmax - k_2 T}{\umax + \wmax + \delta}.
    \end{align}
    Thus, \changed{evaluating $\bar u_{k-1}$ and $\bar u_k$ using} \eqref{eq:implc:u} yields
    \begin{align}
        \bar u_{k-1} &= \bar u_k - k_1 ( \spowf{z_k}{1}{2} - \spowf{z_{k+1}}{1}{2} ) \nonumber \\
        &\quad- 2 (v_{k+1} - v_k) + (v_k - v_{k-1}) \nonumber \\
        &\ge \bar u_k + \frac{k_1^2 T}{2} \frac{\umax - \wmax - k_2 T}{\umax + \wmax + \delta}  - k_2 T\nonumber \\
        &\quad - (v_{k+1} - v_k) + (v_k - v_{k-1}) \nonumber \\
        \label{eq:ukm1}
        & = \bar u_k + T \epsilon + \gamma_k
    \end{align}
    with the abbreviation $\gamma_k = (v_k - v_{k-1}) - (v_{k+1} - v_k)$.
    Consequently, $\bar u_{k-1} \ge \umax - 2 k_2 T$.
    Distinguish cases $v_{k-1} < \umax - 2 k_2 T$ and $v_{k-1} \ge \umax - 2 k_2 T$.
    In the first case, it will be shown that $v_k - v_{k-1} \ge k_2 T$, which implies $\gamma_k \ge 0$ and allows to conclude $\bar u_{k-1} \ge \bar u_k + T \epsilon$ from \eqref{eq:ukm1}.
    To see this, suppose to the contrary that $v_k - v_{k-1} = c k_2 T$ with some $c < 1$.
    Then, $\bar u_{k-1} \ge \umax - (1 - c) k_2 T$ \changed{and $2 v_k - v_{k-1} - u_{k-1} \ge 0$ follow from \eqref{eq:ukm1} and \eqref{eq:implc:v}, respectively.}
    \changed{The former implies $u_{k-1} = \sat_U(\bar u_{k-1}) \ge U - (1-c) k_2 T$ and the latter yields $v_k \ge (u_{k-1} + v_{k-1})/2$, leading to the contradiction}
\begin{align}
        v_k - v_{k-1} &\ge \changed{\frac{u_{k-1} - v_{k-1}}{2} \ge \frac{- (1-c) k_2 T + 2 k_2 T }{2}} \nonumber \\
        &= \frac{1+c}{2} k_2 T > c k_2 T.
    \end{align}
    In the second case, the relation $\bar u_{k-1} \ge \umax - 2 k_2 T$ implies $u_{k-1} \ge \umax - 2 k_2 T$, and \eqref{eq:implc:explicit:v} yields the inequality $v_{k} \ge \min\{ v_{k-1}, u_{k-1}\} \ge \umax - 2 k_2 T$.
\changed{Since $u_k = U$ and $u_{k-1}, v_k, v_{k-1} \in [U- 2k_2 T, U]$, one may conclude $|v_{k-1} - u_{k-1}| \le 2 k_2 T$ and $|v_k - u_k| \le 2 k_2 T$}.
    \changed{By applying} \eqref{eq:implc:explicit:v} \changed{three times}, $\gamma_k$ may \changed{then} be bounded as
    \begin{align}
        \gamma_k &= \frac{u_{k-1} - v_{k-1}}{2} - \frac{u_k - v_k}{2} = \frac{u_{k-1} - \umax}{2} + \frac{v_{k} - v_{k-1}}{2} \nonumber \\
        &= \frac{u_{k-1} - \umax}{2} + \frac{u_{k-1} - v_{k-1}}{4} \ge \frac{3}{4} (u_{k-1} - \umax).\hspace{-1.15ex}
    \end{align}
    If $u_{k-1} = \umax$, then $\gamma_k \ge 0$ and $\bar u_{k-1} \ge \bar u_k + T \epsilon$ follows from \eqref{eq:ukm1}. Otherwise, $u_{k-1} = \bar u_{k-1}$, leading to
    \begin{equation}
        \bar u_{k-1} \ge \bar u_k + T \epsilon + \frac{3}{4} \bar u_{k-1} - \frac{3}{4} \umax \ge \frac{\bar u_k}{4} + T \epsilon + \frac{3}{4} \bar u_{k-1}.
    \end{equation}
    \changed{in \eqref{eq:ukm1};} solving for $\bar u_{k-1}$ yields $\bar u_{k-1} \ge \bar u_k + 4 T \epsilon$.
\qed

\endgroup

\bibliographystyle{elsarticle-harv}
\bibliography{literature}           

\end{document}